\documentclass[conference,final]{IEEEtran}

\usepackage[final]{graphicx}
\usepackage[reqno]{amsmath}
\usepackage{amssymb}
\usepackage{cite}

\newfont{\bbb}{msbm10 scaled 500}

\newfont{\bb}{msbm10 scaled 1100}

\newcommand{\xv}{{\bf x}}
\newcommand{\yv}{{\bf y}}

\newcommand{\Xm}{{\bf X}}
\newcommand{\Ym}{{\bf Y}}

\newcommand{\Ac}{{\cal A}}

\newcommand{\Cc}{{\cal C}}

\newcommand{\Nc}{{\cal N}}
\newcommand{\Oc}{{\cal O}}
\newcommand{\Pc}{{\cal P}}
\newcommand{\Qc}{{\cal Q}}
\newcommand{\Rc}{{\cal R}}
\newcommand{\Sc}{{\cal S}}

\newcommand{\Wc}{{\cal W}}

\newcommand{\Xc}{{\cal X}}
\newcommand{\Yc}{{\cal Y}}

\newtheorem{theorem}{Theorem}
\newtheorem{corollary}[theorem]{Corollary}
\newtheorem{definition}[theorem]{Definition}

\IEEEoverridecommandlockouts

\author{
    \authorblockN{Onur Ozan Koyluoglu and Hesham El~Gamal
    \thanks{Hesham El Gamal also serves as the Director for the
    Wireless Intelligent Networks Center (WINC), Nile University,
    Cairo, Egypt.}\\}
    \authorblockA{
    Department of Electrical and Computer Engineering\\
    The Ohio State University\\
    Columbus, OH 43210, USA\\
    Email: \{koyluogo,helgamal\}@ece.osu.edu}\\}

\title{On the Secrecy Rate Region for the Interference Channel}

\begin{document}

\maketitle


\begin{abstract}

This paper studies interference channels with security constraints.
The existence of an external eavesdropper in a two-user interference channel
is assumed, where the network users would like to secure their messages from
the external eavesdropper. The cooperative binning and channel prefixing scheme
is proposed for this system model which allows users to cooperatively add
randomness to the channel in order to degrade the observations of the external
eavesdropper. This scheme allows users to add randomness to the channel
in two ways: 1) Users cooperate in their design of the binning codebooks,
and 2) Users cooperatively exploit the channel prefixing technique.
As an example, the channel prefixing technique is exploited in the
Gaussian case to transmit a superposition signal consisting of
binning codewords and independently generated noise samples. Gains
obtained form the cooperative binning and channel prefixing scheme
compared to the single user scenario reveals the positive effect of
interference in increasing the network security. Remarkably, interference
can be exploited to cooperatively add randomness into the network in
order to enhance the security.

\end{abstract}


%


\section{Introduction}\label{sec:Introduction}

In this work, we consider two-user interference channels with an external
eavesdropper. Without the secrecy constraints, the interference channel is
studied extensively in the literature. However, the capacity region is still
not known except for some special cases
\cite{Sato:The81,Shang:A07,Annapureddy:A08,Motahari:A08}.
Interference channels with confidential messages is recently studied
by~\cite{Liang:Cognitive07,Liu:Discrete08,Koyluoglu:On08}. Nonetheless, the
external eavesdropper scenario has not been addressed extensively in the
literature yet. In fact, the only relevant work regarding the security of the
interference channels with an external eavesdropper is the study of the secure
degrees of freedom (DoF) in the $K$-user Gaussian interference channels under
frequency selective fading models~\cite{Koyluoglu:On08}, where it is shown
that positive secure DoFs are achievable for each user in the network.

In this work, we propose the cooperative binning and channel prefixing scheme
for (discrete) memoryless interference channels with an external eavesdropper.
The proposed scheme allows for cooperation in adding randomness to the channel
in two ways: $1$) Cooperative binning: The random binning technique
of~\cite{Wyner:Wiretap75} is cooperatively exploited at both users.
$2$) Channel prefixing: Users exploit the channel prefixing technique
of~\cite{Csiszar:Broadcast78} in a cooperative manner.
The proposed scheme also utilizes the message-splitting technique
of~\cite{Han:A81} and partial decoding of the interfering signals is made
possible at the receivers. The achievable secrecy rate region with
the proposed scheme is given. For the Gaussian interference channel, the
channel prefixing technique is exploited to inject artificially generated noise
samples into the network, where we also allow power control at transmitters to
enhance the security of the network.

The proposed scheme is closely related with that
of~\cite{Lai:The,Tekin:The08,Negi:Secret05}.\cite{Lai:The} considered the
relay-eavesdropper channel and proposed the noise-forwarding scheme where the
relay node sends a codeword from an independently generated codebook to add
randomness to the network in order to enhance the security of the main channel.
\cite{Tekin:The08} considered Gaussian multiple-access wire-tap channels
and proposed the cooperative jamming scheme in which users transmit their
codewords or add randomness to the channel by transmitting noise samples,
but not both. The approach in this sequel, when specialized to the Gaussian
multiple access channel with an external eavesdropper, generalizes and extends
the proposed achievable regions given in~\cite{Tekin:The08}, due to the
implementation of simultaneous cooperative binning and jamming at the
transmitters together with more general time-sharing approaches.
This simultaneous transmission of secret messages and noise samples from
transmitters is considered by~\cite{Negi:Secret05}. In~\cite{Negi:Secret05},
authors proposed artificially generated noise injection schemes for multi-transmit
antenna wire-tap channels, in which the superposition of a secrecy signal and
an artificially generated noise is transmitted from the transmitter, where the
noisy transmission only degrades the eavesdropper's channel. For the single
transmit antenna case, wire-tap channels with helper nodes is considered,
in which helper nodes transmit artificially generated noise samples
in order to degrade the eavesdropper's channel. Remarkable, exploitation of
the channel prefixing technique was transparent in these previous studies.
The proposed scheme in this work shows that the benefit of cooperative jamming
scheme of~\cite{Tekin:The08} and noise injection scheme of~\cite{Negi:Secret05}
originates from the channel prefixing technique.
In addition, compared to~\cite{Liu:Discrete08}, the proposed scheme allows for
cooperation via \emph{both} binning and channel prefixing techniques,
whereas in~\cite{Liu:Discrete08} one of the transmitters is allowed to generate
and transmit noise together with the secret signal and cooperation among
network users as considered in this sequel was not implemented for the confidential
message scenario.

The rest of this work is organized as follows. Section II introduces
the system model. In Section III, the main result for discrete
memoryless interference channels is given. Section IV is devoted to
some examples of the proposed scheme for Gaussian channels. Finally,
we provide some concluding remarks in Section V.


\section{System Model}
\label{sec:Model}

We consider a two-user interference channel with an external eavesdropper
(IC-EE), comprised of two transmitter-receiver pairs and an additional
eavesdropping node. The discrete memoryless IC-EE is denoted by
$$(\Xc_1 \times \Xc_2, p(y_1,y_2,y_e|x_1,x_2),
\Yc_1 \times \Yc_2 \times \Yc_e),$$ for some finite sets
$\Xc_1, \Xc_2, \Yc_1, \Yc_2, \Yc_e$.
Here the symbols $(x_1,x_2)\in \Xc_1 \times \Xc_2$ are the
channel inputs and the symbols $(y_1,y_2,y_e)\in \Yc_1 \times \Yc_2
\times \Yc_e$ are the channel outputs observed at the
decoder $1$, decoder $2$, and at the eavesdropper, respectively.
The channel is memoryless and time-invariant:
\footnote{In this work, we have the following notation:
Vectors are denoted as $\xv^i=\{x(1),\cdots,x(i)\}$, where
we omit the $i$ if $i=n$, i.e., $\xv=\{x(1),\cdots,x(n)\}$.
Random variables are denoted with capital letters ($X$),
and random vectors are denoted as bold-capital letters ($\Xm^i$).
Again, we drop the $i$ for $\Xm=\{X(1),\cdots,X(n)\}$.
Lastly, $[x]^+\triangleq\max\{0,x\}$,
$\bar{\alpha}\triangleq 1-\alpha$, and
$\gamma(x)\triangleq\frac{1}{2}\log_2(1+x)$.}
$$p(y_1(i),y_2(i),y_e(i)|\xv_1^i,\xv_2^i,
\yv_1^{i-1},\yv_2^{i-1},\yv_e^{i-1})$$
$$= p(y_1(i),y_2(i),y_e(i)|x_1(i),x_2(i)).$$
We assume that each transmitter $k\in\{1,2\}$ has a secret message
$W_k$ which is to be transmitted to the respective receivers in $n$
channel uses and to be secured from the external eavesdropper.
In this setting,
an $(n,M_1,M_2,P_{e,1},P_{e,2})$ secret codebook has the following components:

$1$) The secret message sets $\Wc_k=\{1,...,M_k\}$ for
transmitter $k=1,2$.

$2$) Encoding function $f_k(.)$ at transmitter $k$
which map the secret messages to the transmitted symbols, i.e.,
$f_k:w_k\to \Xm_k$ for each $w_k\in\Wc_k$ for $k=1,2$.

$3$) Decoding function $\phi_k(.)$ at receiver $k$
which map the received symbols to estimate of the message:
$\phi_k(\Ym_k)=\hat{w}_k$ for $k=1,2$.

Reliability of the transmission of user $k$ is measured by $P_{e,k}$, where
\begin{eqnarray}
P_{e,k}\triangleq\frac{1}{M_1 M_2}\sum\limits_{(w_1,w_2)\in \Wc_1\times\Wc_2}
\textrm{Pr} \left\{\phi_k(\Ym_k)\neq w_k | E_{w_1,w_2}\right\},   \nonumber
\end{eqnarray}
where $E_{w_1,w_2}$ is the event that $(w_1,w_2)$ is transmitted from the
transmitters.

For the secrecy requirement, the level of ignorance of the eavesdropper
with respect to the secured messages is measured by the equivocation
rate
$$\frac{1}{n}H\left(W_1,W_2|\Ym_e\right).$$
We say that the rate tuple $(R_1,R_2)$ is achievable for the IC-EE
if, for any given $\epsilon > 0$, there exists an
$(n,M_1=2^{nR_1},M_2=2^{nR_2},P_{e,1},P_{e,2})$ secret codebook such that,
\begin{eqnarray}
\max\{P_{e,1},P_{e,2}\} \leq \epsilon \nonumber,
\end{eqnarray} and
\begin{eqnarray}\label{eq:secrecy}
R_1+R_2 -\frac{1}{n}H\left(W_1,W_2|\Ym_e\right) &\leq& \epsilon
\end{eqnarray}
for sufficiently large $n$. The secrecy capacity region is the closure of
the set of all achievable rate pairs $(R_1,R_2)$ and is denoted as
$\mathbb{C}^{\textrm{IC-EE}}$.

\subsection{The Gaussian Interference Channel with an External Eavesdropper
in Standard Form}
\label{sec:StandardForm}

The Gaussian interference channel in standard form is given
in~\cite{Carleial:Interference78}. We have the same transformation here
for the Gaussian interference channel with an external eavesdropper (GIC-EE)
model. We remark that the channel capacity will remain
the same as the transformations are invertible.
We represent the average power constraints of the transmitters as
$P_k$, where codewords should
satisfy $\frac{1}{n}\sum\limits_{t=1}^n
(X_k(t))^2\leq P_k$ for $k=1,2$. Here
the input-output relationship, i.e., $p(y_1,y_2,y_e|x_1,x_2)$,
changes to the following:
\begin{eqnarray}\label{eq:GIC-EE}
Y_1 & = & X_1 + \sqrt{c_{21}} X_2 + N_1 \nonumber\\
Y_2 & = & \sqrt{c_{12}} X_1 + X_2 + N_2 \\
Y_e & = & \sqrt{c_{1e}} X_1 + \sqrt{c_{2e}} X_2 + N_e \nonumber,
\end{eqnarray}
where $N_{k}\sim\Nc(0,1)$ for $k=1,2,e$ as depicted in Fig.~$1$.
The secrecy capacity region of the GIC-EE is denoted as
$\mathbb{C}^{\textrm{GIC-EE}}$.

\begin{figure}[htb] 
    \centering
    \includegraphics[width=0.9\columnwidth]{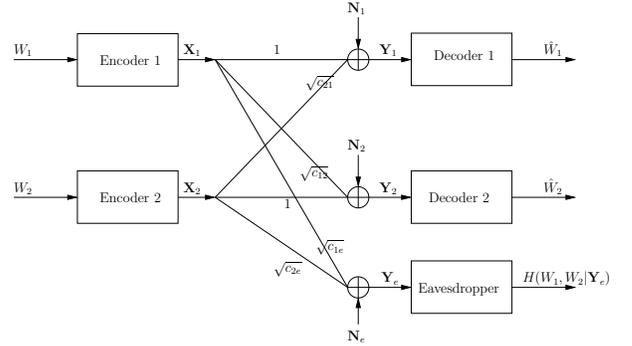}
    \caption{The Gaussian interference channel
    with an external eavesdropper in standard form.
    $N_{k}\sim\Nc(0,1)$ for $k=1,2,e$.}
\end{figure}


\section{The Discrete Memoryless Interference Channel
with an External Eavesdropper}
\label{sec:TheDiscreteMemorylessInterferenceChannelwithanExternalEavesdropper}

In this section, we introduce the proposed cooperative binning and
channel prefixing scheme for the IC-EE model. With this scheme,
transmitters design their secrecy codebooks using the random binning
technique~\cite{Wyner:Wiretap75}. This binning structure in the
codebook let a transmitter to add randomness in its own signals.
However, the price of adding extra randomness to secure the transmission
appear as a rate loss in the achievable rate expressions. In our scenario,
the proposed strategy allows for cooperation in design of these binning
codebooks, and allows for cooperation in prefixing the channel as we utilize
the channel prefixing technique of~\cite{Csiszar:Broadcast78} at both users.
Hence, users of the interference channel will add only
\emph{sufficient} amount of randomness as the other
user will help to increase the randomness seen by the eavesdropper. The
achievable secure rate region with this scheme is described below.

First consider auxiliary random variables $Q$, $C_1$, $S_1$, $O_1$, $C_2$,
$S_2$, and $O_2$ defined on arbitrary finite sets $\Qc$, $\Cc_1$, $\Sc_1$, $\Oc_1$,
$\Cc_2$, $\Sc_2$, and $\Oc_2$, respectively.
Now, let $\Pc$ be the set of all joint distributions of the random variables
$Q$, $C_1$, $S_1$, $O_1$, $C_2$, $S_2$, $O_2$, $X_1$, $X_2$, $Y_1$, $Y_2$,
and $Y_e$ that factors as
$p(q,c_1,s_1,o_1,c_2,s_2,o_2,x_1,x_2,y_1,y_2,y_e)$ $=$
$p(q)$
$p(c_1|q)$ $p(s_1|q)$ $p(o_1|q)$
$p(c_2|q)$ $p(s_2|q)$ $p(o_2|q)$
$p(x_1|c_1,s_1,o_1,q)$
$p(x_2|c_2,s_2,o_2,q)$
$p(y_1,y_2,y_e|x_1,x_2)$.
Here, the variable $Q$ serves as a time-sharing parameter.
  See, for example,~\cite{Han:A81,ThomasAndCover} for a discussion
  on time-sharing parameters.
 The variable $C_1$ is used to construct the \emph{common} secured
  signal of transmitter $1$ that has to be decoded at both receivers, where
  the random binning technique of~\cite{Wyner:Wiretap75} is used for this
  construction.
 The variable $S_1$ is used to construct the \emph{self} secured signal
  that has to be decoded at receiver $1$ but not at receiver $2$, where
  the random binning technique of~\cite{Wyner:Wiretap75} is used for
  this construction.
 The variable $O_1$ is used to construct \emph{other} signal of
  transmitter $1$ that has to be decoded at receiver $2$ but not at
  receiver $1$, where the conventional random codebook construction, see
  for example~\cite{ThomasAndCover}, is used for this signal, i.e., no
  binning is implemented.
 Similarly, $C_2$, $S_2$, and $O_2$ are utilized at user $2$.
 Finally, it is important to remark that the channel prefixing
  technique of~\cite{Csiszar:Broadcast78} is exploited with this
  construction as we transformed the channel $p(y_1,y_2,y_e|x_1,x_2)$
  to $p(y_1,y_2,y_e|c_1,s_1,o_1,c_2,s_2,o_2,q)$ using the prefixes
  $p(x_1|c_1,s_1,o_1,q)$ and $p(x_2|c_2,s_2,o_2,q)$.

To ease the presentation, we first state the following definitions.
We define $T_1\triangleq C_1$, $T_2\triangleq S_1$, $T_3\triangleq O_1$,
$T_4\triangleq C_2$, $T_5\triangleq S_2$, $T_6\triangleq O_2$ and
corresponding rates $R_{T_i}$ and $R_{T_i}^x$. Note that we choose
$R_{O_1}=R_{O_2}=0$ below. Also, we define
$T_{\Sc}\triangleq\{T_i|i\in\Sc\}$.
\begin{definition}
$\Rc_1(p)$ is the set of all tuples
$(R_{C_1},R_{C_1}^x,R_{S_1},R_{S_1}^x,R_{C_2},R_{C_2}^x,R_{O_2}^x)$ satisfying
\begin{eqnarray}
\sum\limits_{i\in\Sc} R_{T_i} + R_{T_i}^x \leq I(T_{\Sc};Y_1|T_{\Sc^c},Q),
\forall \Sc \subseteq \{1,2,4,6\},
\end{eqnarray}
for a given joint distribution $p$.
\end{definition}
\begin{definition}
$\Rc_2(p)$ is the set of all tuples
$(R_{C_2},R_{C_2}^x,R_{S_2},R_{S_2}^x,R_{C_1},R_{C_1}^x,R_{O_1}^x)$ satisfying
\begin{eqnarray}
\sum\limits_{i\in\Sc} R_{T_i} + R_{T_i}^x \leq I(T_{\Sc};Y_2|T_{\Sc^c},Q),
\forall \Sc \subseteq \{1,3,4,5\},
\end{eqnarray}
for a given joint distribution $p$.
\end{definition}
\begin{definition}
$\Rc_e(p)$ is the set of all tuples
$(R_{C_1}^x,R_{S_1}^x,R_{O_1}^x,R_{C_2}^x,R_{S_2}^x,R_{O_2}^x)$ satisfying
\begin{eqnarray}
\sum\limits_{i\in\Sc} R_{T_i}^x &\leq& I(T_{\Sc};Y_e|T_{\Sc^c},Q),
\forall \Sc \subsetneqq \{1,\cdots,6\}, \nonumber\\
\sum\limits_{i\in\{1,2,3,4,5,6\}} R_{T_i}^x &=&
I(T_1,T_2,T_3,T_4,T_5,T_6;Y_e|Q),
\end{eqnarray}
for a given joint distribution $p$.
\end{definition}

\begin{definition}
$\Rc(p)$ is the closure of all $(R_1,R_2)$ satisfying
\begin{eqnarray}
R_1&=&R_{C_1}+R_{S_1}, \nonumber\\
R_2&=&R_{C_2}+R_{S_2}, \nonumber\\
(R_{C_1},R_{C_1}^x,R_{S_1},R_{S_1}^x,R_{C_2},R_{C_2}^x,R_{O_2}^x) &\in& \Rc_1(p), \nonumber\\
(R_{C_2},R_{C_2}^x,R_{S_2},R_{S_2}^x,R_{C_1},R_{C_1}^x,R_{O_1}^x) &\in& \Rc_2(p), \nonumber\\
(R_{C_1}^x,R_{S_1}^x,R_{O_1}^x,R_{C_2}^x,R_{S_2}^x,R_{O_2}^x) &\in& \Rc_e(p), \nonumber
\end{eqnarray}
and
\begin{eqnarray}\label{eq:RI}
R_{C_1}\geq 0, R_{C_1}^x\geq 0, R_{S_1}\geq 0, R_{S_1}^x\geq 0, R_{O_1}^x\geq 0,\nonumber\\
R_{C_2}\geq 0, R_{C_2}^x\geq 0, R_{S_2}\geq 0, R_{S_2}^x\geq 0, R_{O_2}^x\geq 0,
\end{eqnarray}
for a given joint distribution $p$.
\end{definition}

We now state the main result of the paper.
The achievable secrecy rate region using the cooperative
binning and channel prefixing scheme is as follows.
\begin{theorem}\label{thm:R}
$\Rc^{\textrm{IC-EE}}
\triangleq
\textrm{ the closure of }
\left\{
\bigcup\limits_{p\in\Pc} \Rc(p)
\right\}
\subset \mathbb{C}^{\textrm{IC-EE}}$.
\end{theorem}
\begin{proof}
The proof is omitted and will be provided
in the journal version of this work.
\end{proof}


\section{The Gaussian Interference Channel with an External Eavesdropper}
\label{sec:TheGaussianInterferenceChannelwithanExternalEavesdropper}

In this section, we provide some examples of the proposed coding scheme for
Gaussian interference channels and show that the proposed scheme provides
gains in securing the network by exploiting cooperative binning,
cooperative channel prefixing, and time-sharing techniques.

Firstly, we describe how the channel prefixing can be implemented in this
Gaussian scenario. Here, one can independently generate and transmit
noise samples for each channel use from the transmitters (without
constructing a codebook and sending one of its messages) to enhance
the security of the network. As there is no design of a
codebook at the interfering user for this noise transmission, receivers
and the eavesdropper can only consider this transmission as noise.
Accordingly, transmitter $k\in\{1,2\}$ uses power $P_k^b$ for the construction
of its (binning) codewords, which are explained in the previous section,
and obtains, somehow, the signal $X_k^b\sim\Nc(0,P_k^b)$. In addition, it
uses power $P_k^j$ for its jamming signal and generates i.i.d. noise samples
represented by $X_k^j\sim\Nc(0,P_k^j)$, where we choose
$P_k^b+P_k^j \leq P_k$. Then, it sends $X_k^b+X_k^j$ to the channel,
instead of just sending $X_k^b$.

Now, we can use the scheme proposed in the previous section for the design
of the signals $X_k^b$. Below we will use superposition coding
to construct this signal. But first,
for a rigorous presentation, we provide some definitions.
Let $\Ac$ denote the set of all tuples
$\left(P_1^c(q),P_1^s(q),P_1^o(q),P_2^c(q),P_2^s(q),P_2^o(q),P_1^j(q),P_2^j(q)\right)$
satisfying
$P_k^b(q) \triangleq P_k^c(q)+P_k^s(q)+P_k^o(q)$ and
$\sum\limits_{q\in\Qc}(P_k^b(q)+P_k^j(q))p(q) \leq P_k,$
for $k=1,2.$

Now, we define a set of joint distributions
$\Pc_1$ as follows.

$\Pc_1  \triangleq  \big\{p \: | \: p\in\Pc,$
$(P_1^c(q),P_1^s(q),P_1^o(q),P_2^c(q),P_2^s(q),P_2^o(q),P_1^j(q),P_2^j(q))\in\Ac,$
$C_1\sim\Nc(0,P_1^c(q))$, $S_1\sim\Nc(0,P_1^s(q))$, $O_1\sim\Nc(0,P_1^o(q))$,
$C_2\sim\Nc(0,P_2^c(q))$, $S_2\sim\Nc(0,P_2^s(q))$, $O_2\sim\Nc(0,P_2^o(q))$,
$X_1^j\sim\Nc(0,P_1^j(q))$, $X_2^j\sim\Nc(0,P_2^j(q))$,
$X_1=C_1+S_1+O_1+X_1^j$, $X_2=C_2+S_2+O_2+X_2^j \big\}$,
where the Gaussian model given in (\ref{eq:GIC-EE})
gives $p(y_1,y_2,y_e|x_1,x_2)$.

Then, the following region is achievable
for the Gaussian interference channel with an external eavesdropper.
\begin{corollary}\label{thm:RG}
$\Rc^{\textrm{GIC-EE}}
\triangleq
\textrm{ the closure of }
\bigg\{
\bigcup\limits_{p\in\Pc_1} \Rc(p)
\bigg\}
\subset  \mathbb{C}^{\textrm{GIC-EE}}$.
\end{corollary}

We emphasize the way of implementing the channel prefixing technique
of~\cite[Lemma $4$]{Csiszar:Broadcast78}: $p(x_k|c_k,s_k,o_k,q)$ is
chosen by $X_k=C_k+S_k+O_k+X_k^j$. With this choice, we are able to
implement simultaneous binning and jamming at the transmitters
together with a power control.

\subsection{Subregions of $\Rc^{\textrm{GIC-EE}}$}

We now present a computationally simpler region.
Consider $\Pc_2\triangleq \{p \: | \: p\in\Pc_1, Q=\emptyset\}.$
\begin{corollary}\label{thm:RG2}
$\Rc_2^{\textrm{GIC-EE}}
\triangleq
\textrm{convex closure of }
\bigg\{
\bigcup\limits_{p\in\Pc_2} \Rc(p)
\bigg\}$
$\subset \Rc^{\textrm{GIC-EE}} \subset  \mathbb{C}^{\textrm{GIC-EE}}$.
\end{corollary}

We also provide a sub-region of $\Rc_2^{\textrm{GIC-EE}}$ that will be
used for numerical results. Define a set of joint distributions $\Pc_3$.
$$\Pc_3\triangleq \{p \: | \: p\in\Pc_2, P_1^s=P_1^o=P_2^s=P_2^o=0 \}.$$
\begin{corollary}\label{thm:RG3}
$\Rc_3^{\textrm{GIC-EE}}
\triangleq
\textrm{convex closure of }
\bigg\{
\bigcup\limits_{p\in\Pc_3} \Rc(p)
\bigg\}$
$\subset \Rc^{\textrm{GIC-EE}} \subset  \mathbb{C}^{\textrm{GIC-EE}}$.
\end{corollary}

It is important to note that we use the convex closure of the rate regions
instead of using a time-sharing parameter in these subregions.
We have already given the more general region above and we
conjecture that it is possible to extend these achievable subregions by a
different choice of channel prefixing or by using a time-sharing approach.

Accordingly, we consider a TDMA-like approach, which will show that even a simple
type of time-sharing is beneficial. Here we divide the $n$ channel uses
into two intervals of lengths represented by $\alpha n$ and $(1-\alpha)n$, where
$0\leq \alpha \leq 1$ and $\alpha n$ is assumed to be an integer. The
first period, of length $\alpha n$, is dedicated to secure transmission
for user $1$. During this time, transmitter $1$ generates binning codewords
using power $P_1^b$ and jams the channel using power $P_1^{j_1}$; and
transmitter $2$ jams the channel using power $P_2^{j_1}$. For the
second period the roles of the users are reversed, where users use powers
$P_2^b$, $P_2^{j_2}$, and $P_1^{j_2}$. We call this scheme cooperative TDMA
(C-TDMA) and obtain the following region in this case.

\begin{corollary}\label{thm:RC-TDMA}
$\Rc_{C-TDMA} \subset \Rc^{\textrm{GIC-EE}} \subset
\mathbb{C}^{\textrm{GIC-EE}}$, where
$\Rc_{C-TDMA}\triangleq \textrm{ closure of the convex hull of }$
\begin{eqnarray}
\left\{
\mathop{
\mathop{\bigcup\limits_{0\leq \alpha \leq 1}}
\limits_{
\alpha \left(P_1^b+P_1^{j_1}\right)+ \bar{\alpha} P_1^{j_2} \leq P_1}}
\limits_{\alpha P_2^{j_1} + \bar{\alpha}\left(P_2^b+P_2^{j_2}\right) \leq P_2}
(R_1,R_2)
\right\},
\end{eqnarray}
where
\begin{eqnarray}
R_1=\alpha \left[
\gamma\left(\scriptstyle{\frac{P_1^b}{1+P_1^{j_1}+c_{21}P_2^{j_1}} }\right)
-{}\gamma\left(\scriptstyle{\frac{c_{1e}P_1^b}{1+c_{1e}P_1^{j_1}
+c_{2e}P_2^{j_1}} }\right)
\right]^+,\nonumber
\end{eqnarray}
and
\begin{eqnarray}
R_2=\bar{\alpha} \left[
\gamma\left( \scriptstyle{\frac{P_2^b}{1+P_2^{j_2}+c_{12}P_1^{j_2}} } \right)
-{}\gamma\left( \scriptstyle{\frac{c_{2e}P_2^b}{1+c_{2e}P_2^{j_2}
+c_{1e}P_1^{j_2}} }\right)
\right]^+.\nonumber
\end{eqnarray}
\end{corollary}

Note that, we only consider adding randomness by noise injection for the
cooperative TDMA scheme above. However, our coding scheme presented in
the previous section allows for an implementation of more general
cooperation strategies, in which users can add randomness to the channel
in two ways: adding randomness via cooperative binning and adding
randomnees via cooperative channel prefixing. A user by implementing
\emph{both} of these approaches can help the other one in a time-division
setting. We again remark that the proposed cooperative binning and channel
prefixing scheme allows even more general approaches such as having more
than two time-sharing periods.

\subsection{Numerical Results and Discussion}
In this section we provide numerical results for the following
subregions of the achievable region given by Corollary~\ref{thm:RG}.

  $1$) $\Rc_3^{\textrm{GIC-EE}}$: This region is provided above, where
  we utilize both cooperative binning and channel prefixing.

  $2$) $\Rc_3^{\textrm{GIC-EE}}{\textrm{(b or cp)}}$: Here we utilize
  either cooperative binning or channel prefixing scheme
  at a transmitter, but not both.

  $3$) $\Rc_3^{\textrm{GIC-EE}}{\textrm{(ncp)}}$: Here we only utilize
  cooperative binning. Accordingly, jamming powers are set to zero.

  $4$) $R_{C-TDMA}$: This region is an example of utilizing both
  time-sharing and cooperative channel prefixing. No cooperative binning
  is used.

  $5$) $R_{C-TDMA}{\textrm{(nscp)}}$: Here we do not allow transmitters
  to jam the channel during their dedicated time slots and call this case
  no self channel prefixing (nscp).

  $6$) $R_{C-TDMA}{\textrm{(ncp)}}$: Here no channel prefixing is implemented.
  This case refers to conventional TDMA scheme, in which users are allowed to
  transmit during only their assigned slots. Hence, this scheme only utilizes
  time-sharing.

Numerical results are provided in Fig.~$2$ and Fig.~$3$. The first
scenario depicted in Fig.~$2$ shows the benefits of cooperative binning
technique. Also, cooperative channel prefixing does not help to enlarge
the secure rate region in this scenario. Secondly, in Fig.~$3$, we
consider an asymmetric scenario, in which the first user has a weak
channel to the eavesdropper but the second user has a strong channel to the
eavesdropper. Here, the second user can help the first one to increase its
secrecy rate. However, channel prefixing and time-sharing does not
help to the second user as it can not achieve positive secure rate
without an implementation of cooperative binning. Remarkable, cooperative
binning technique helps the second user to achieve positive secure transmission
rate in this case. These observations suggest the implementation of all
three techniques (cooperative binning, cooperative channel prefixing, and time-sharing)
as considered in our general rate region, i.e., $\Rc^{\textrm{GIC-EE}}$.

\begin{figure}[htb] 
    \centering
    \includegraphics[width=0.7\columnwidth]{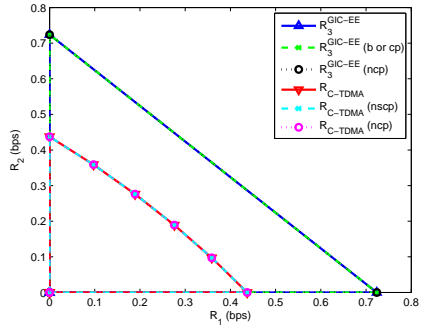}
    \caption{Numerical results for GIC-EE with $c_{12}=c_{21}=1.9$,
    $c_{1e}=c_{2e}=0.5$, $P_{1}=P_{2}=10$.}
\end{figure}

\begin{figure}[htb] 
    \centering
    \includegraphics[width=0.7\columnwidth]{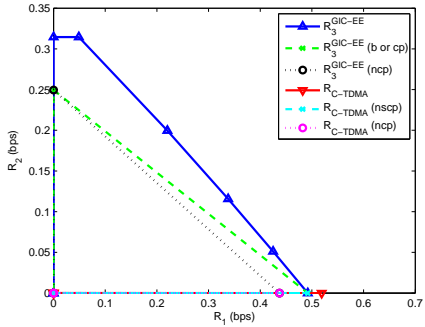}
    \caption{Numerical results for GIC-EE with $c_{12}=1.9$,
    $c_{21}=1$, $c_{1e}=0.5$, $c_{2e}=1.6$, $P_{1}=P_{2}=10$.}
\end{figure}

\subsection{Some Implications of the Proposed Scheme}
It can be shown that the proposed scheme reduces to the noise forwarding
scheme of~\cite{Lai:The} for the discrete memoryless relay-eavesdropper
channel. Remarkable, the channel prefixing technique can be exploited in this
scenario to increase the achievable secure rates. For example, for the
Gaussian channel, injecting i.i.d. noise samples can increase the achievable
secure transmission rates as shown in~\cite{Tang:The08}. Our result here
shows that the gain resulting from the noise injection comes from the
exploitation of the channel prefixing technique. In addition, the proposed
scheme, when specialized to a Gaussian multiple-access scenario, results in
an achievable region that generalizes and extends the proposed regions given
in~\cite{Tekin:The08} due to the implementation of simultaneous cooperative
binning and channel prefixing at the transmitters together with more general
time-sharing approaches.


\section{Conclusion}
\label{sec:Conclusion}

In this work, we have considered two-user interference channels with an
external eavesdropper. We have proposed the cooperative binning and channel
prefixing scheme that utilizes random binning, channel prefixing, and
time-sharing techniques and allows transmitters to cooperate in adding
randomness to the channel. For Gaussian interference channels, the
channel prefixing technique is exploited by letting users to inject
independently generated noise samples to the channel. The most
interesting aspect of our results is, perhaps, the unveiling of the
role of interference in cooperatively adding randomness to the channel to
increase the secrecy rates of multi-user networks.




\end{document}